\documentclass[conference]{IEEEtran}

\usepackage{booktabs} 
\usepackage{amsmath, amsthm, amssymb, bbm, setspace, mathrsfs, color, listings, graphicx, multicol, dcolumn, dsfont}
\usepackage[font=footnotesize,labelfont=bf]{caption}
\usepackage{float,subcaption}
\usepackage[left=0.7in,right=1.0in,top=0.7in,bottom=1.0in,nohead]{geometry}
\usepackage[linesnumbered,ruled]{algorithm2e}
\usepackage[numbers]{natbib}

\newtheorem{thm}{Theorem}

\newtheorem{property}{Property}

\newtheorem{prop}[thm]{Proposition}

\newcommand{\utwi}[1]{\mbox{\boldmath $ #1$}}
\newcommand{\Y}{{\utwi{Y}}}
\newcommand{\X}{{\utwi{X}}}

\DeclareMathOperator*{\argmax}{argmax}

\begin{document}
\title{Pruning and Nonparametric\\ Multiple Change Point Detection}

\author{\IEEEauthorblockN{Wenyu Zhang}
\IEEEauthorblockA{Department of Statistical Science\\
Cornell University\\
Email: wz258@cornell.edu}
\and
\IEEEauthorblockN{Nicholas A. James}
\IEEEauthorblockA{Department of Operations Research\\ and Information Engineering\\
Cornell University\\
Email: nj89@cornell.edu}
\and
\IEEEauthorblockN{David S. Matteson}
\IEEEauthorblockA{Departments of Statistical Science\\ and Social Statistics\\
Cornell University\\
Email: matteson@cornell.edu}}


\maketitle

\begin{abstract}
Change point analysis is a statistical tool to identify homogeneity within time series data. We propose a pruning approach for approximate nonparametric estimation of multiple change points. This general purpose change point detection procedure `cp3o' applies a pruning routine within a dynamic program to greatly reduce the search space and computational costs. Existing goodness-of-fit change point objectives can immediately be utilized within the framework. 

We further propose novel change point algorithms by applying cp3o to two popular nonparametric goodness of fit measures: `e-cp3o' uses E-statistics, and `ks-cp3o' uses Kolmogorov-Smirnov statistics. Simulation studies highlight the performance of these algorithms in comparison with parametric and other nonparametric change point methods. Finally, we illustrate these approaches with climatological and financial applications.
\end{abstract}

\section{Introduction}

The analysis of time ordered data, or time series, has become a frequent practice in both academic and industrial settings. 
When analysis is performed it is generally assumed that the data adheres to some form of homogeneity. However, it may not be appropriate, or practical, to apply the same analytical procedure to many different types of time series. The resulting statistical bias from such model misspecification is one of the reasons for the current resurgence of change point analysis, which attempts to partition a time series into homogeneous segments.

A popular approach is to fit the observed data to a parametric model. In this setting a change point corresponds to a change in the monitored parameter(s) \citep{maboudou13,chen11}. Parametric approaches rely heavily upon the assumption that the data behaves according to the predefined distribution model. Otherwise the degree of bias in the obtained results is usually unknown \citep{pitarakis04}. In practice, it is almost always difficult to test for adherence to these assumptions. 

Nonparametric analysis is a natural way to proceed. Since nonparametric approaches make much weaker assumptions than their 
parametric counterparts, they can be used in a much wider variety of settings; for example, the analysis of internet traffic data, where there is no commonly accepted 
distributional model.

In this paper, we address the challenge of designing a customizable procedure that can detect a wide range of changes while appropriately balancing detection accuracy and speed. 
We introduce a new change point search framework called cp3o (\textbf{C}hange \textbf{P}oint \textbf{P}rocedure via \textbf{P}runed \textbf{O}bjectives). 
The cp3o framework is a general purpose search procedure, which means it can be used with a large class of goodness-of-fit metrics to detect change points. For instance, additional knowledge about the data, such as the type of changes which are to be detected, or computational time considerations might direct a user to particular goodness-of-fit metrics. This plug-and-play idea is similar to that in \cite{baranowski16}, such that the users can specify their own goodness-of-fit metrics, or pick from available options based on performance with training data. The cp3o procedure makes use of dynamic programming with search space pruning. This allows the number of change points to be quickly determined, while simultaneously generating all other optimal segmentations as a byproduct.

We further propose two new change point algorithms, named e-cp3o and ks-cp3o, by incorporating two popular nonparametric goodness-of-fit metrics, namely E-statistics and the Kolmogorov-Smirnov statistic, within the cp3o search procedure. 
Results from a variety of simulations show that in most cases the proposed cp3o algorithms provide a good balance between speed and accuracy in comparison with parametric and other nonparametric change point methods.

Both e-cp3o and ks-cp3o algorithms are freely available in the {\tt ecp} {\it R} package on CRAN.

\section{Related Works}

\subsection{Multiple Change Point Detection Methods}

Most existing procedures for performing retrospective multiple change point analysis can be classified as belonging to one of two groups: those that return \emph{approximate} solutions and those that return \emph{exact} solutions. 

Approximate search algorithms tend to rely heavily on a subroutine for finding a single change point. Estimates for multiple change point locations are produced by iteratively applying this subroutine. Examples include binary segmentation and its adaptations such as the Circular Binary Segmentation approach of \cite{olshen04} and the E-Divisive approach of \cite{matteson13}. Approximate procedures tend to produce sub-optimal segmentations of the given time series, but have much lower computational complexity than exact procedures.

Exact search algorithms return segmentations that are optimal with respect to a pre-specified goodness-of-fit metric.
In order to achieve a reasonable computational cost, the utilized goodness-of-fit metrics often satisfy Bellman's Principle of Optimality \citep{bellman52}, and can thus be optimized through the use of dynamic 
programming. Examples of exact algorithms include the Kernel Change Point algorithm, \citep{cappe07} and \citep{arlot12}, and the MultiRank algorithm \citep{fong11}.

\subsection{Pruning Methods}

The runtime of traditional dynamic programming change point detection approaches is still at least quadratic in the length of the time series. However, many of the calculations performed during the dynamic programs do not result 
in the identification of a new change point. These calculations can be viewed as excessive and they quickly compound to slow down analysis. One way to tackle this is by continually pruning the set of potential change point locations. \cite{rigaill10} proposes a 
pruning method that can be used when the goodness-of-fit metric is convex. The PELT method \cite{killick12} is a parametric method which incorporates a pruning step in its dynamic program, such that the expected running time is linear in the length of the time series under certain conditions. However, these methods restrict the options of goodness-of-fit metrics that can be used due to requirements of convexity and parametric objective formulations.

\section{Problem Formulation}
\label{sec: problem}

Let $Z_1, Z_2,\dots, Z_T\in\mathbb R^d$ be a length $T$ sequence of independent 
$d$-dimensional time ordered random variables. We denote $k$ as the true number of change points, where the change points are time indices $1=t_0<t_1<\dots<t_k<t_{k+1}=T+1$, such that 
$Z_i\stackrel{iid}{\sim} F_j$ for $t_j\leq i<t_{j+1}$, and $F_j \neq F_{j+1}$, for distributions $F_j$ with $0\leq j\leq k$.

Given a series of such observations, the challenge is to select the number of change points and change point locations so that the observations within each segment are identically distributed, and the distributions of observations in adjacent segments are different. We approach this problem through the use of goodness-of-fit metrics, which are commonly used for exact search procedures. We also incorporate the parameter $w\geq 1$, which is a user-defined lower bound for the distance between change points.

We refer to a partition of $Z_1, Z_2,\dots, Z_T$ with $\kappa$ segmentation points as a $\kappa$-segmentation. With segmentation points $1 = \tau_0 < \tau_1 < \dots < \tau_\kappa < \tau_{\kappa+1} = T+1$, we quantify the quality of the resulting $\kappa$-segmentation with the empirical goodness-of-fit metric:
\begin{equation*}
    \sum_{j=1}^\kappa 
    \widehat{g}_R\left(\tau_{j-1}, \tau_j, \tau_{j+1}\right)
\end{equation*}

\noindent
where $\widehat{g}_R\left(a, b, c\right) = \widehat{R}\!\left(Z_a^{b-1}, Z_b^{c-1}\right)$ and $Z_a^b=\{Z_i\}_{i=a}^b$, for $a<b<c$. Here $\widehat{R}(\cdot,\cdot)$ is a sample version of a given population measure $R(\cdot,\cdot)$ of the dissimilarity between the distributions of two random variables.

Empirical goodness-of-fit of the $\kappa$-segmentation of a length T sequence with $k$ change points is maximized at 
\begin{align*}
    & \widehat{G}_T(\kappa, w) = 
    \max_{\substack{\tau_1,\tau_2,\dots,\tau_\kappa \\ 
    \tau_i+w\leq\tau_\ell, \ i<\ell \\
    \tau_i \in \{1+w,\dots,T-w+1\}}}
    \sum_{j=1}^{\kappa}
    \widehat{g}_R\left(\tau_{j-1}, \tau_j, \tau_{j+1}\right)
\end{align*}
Calculating $\widehat{G}_T(\kappa, w)$ requires maximization over all $\kappa$-tuples containing a strictly increasing sequence of elements from $1+w$ to $T-w+1$ (that are at least $w$ apart), and hence is computationally expensive. We next introduce an approximation procedure that gives significant speed-ups.

\section{Proposed cp3o Procedure}
\label{sec: procedure}

We adapt the evaluation of $\widehat{G}_T(\kappa, w)$ in two ways to increase computational efficiency:
\begin{itemize}
    \item Approximation of $\widehat{G}_T(\kappa, w)$ to allow the use of dynamic programming,
    \item Pruning to reduce the dynamic program search space.
\end{itemize}

To obtain estimates $\kappa$ and $\{\tau_i\}_{i=1}^\kappa$ for the number of change points $k$ and the change point locations $\{t_i\}_{i=1}^k$, we can calculate $\widehat{G}_T(\kappa, w)$ for a range of values $1\leq \kappa \leq K$ where $K\geq k$ is a user-defined upper bound for $k$, then select $\kappa$ based on a chosen rule which we propose in Section \ref{sec: algorithm}.

\subsection{Dynamic Programming}
\label{sec: dp}

Since there are $\mathcal{O}(T^\kappa)$ possible $\kappa$-segmentations, a direct computation of $\widehat{G}_T(\kappa, w)$ requires $\mathcal{O}(T^\kappa)$ evaluations of the goodness-of-fit metric. 
Instead, we employ dynamic programming in the following fashion. Define
\begin{align*}
&H_t(\kappa, w, \tau) = \widetilde{G}_{\tau-1}(\kappa-1, w) +
    \widehat{g}_R\left( A_{\tau-1}(\kappa-1), \tau, t \right).
\end{align*}

Then, in the $\kappa^{th}$ iteration, for each subsequence $\{Z_i\}_{i=1}^t$ where $1\leq t \leq T$, we define
\begin{align*}
    &A_t(\kappa, w) =\argmax_{\tau\in \{1+\kappa^*w,\dots,t-w+1\}} H_t(\kappa, w, \tau),\\
    &\widetilde{G}_t(\kappa, w) =
    \max_{\tau\in \{1+\kappa^*w,\dots,t-w+1\}} H_t(\kappa, w, \tau),
\end{align*}

\noindent
where $\widetilde{G}_t(\kappa, w)$ denotes the approximation of the optimal goodness-of-fit for the length $t$ subsequence with $\kappa$ segmentation points, and 
$A_t(\kappa, w)$ denotes the location of the $\kappa^{th}$ segmentation point in this approximation. 

$\widetilde{G}_t(\kappa, w)$ is obtained by optimizing over all possible candidates for the $\kappa^{th}$ segmentation point, and approximating the previous $\kappa-1$ change points through $A$. For example, if $\tau$ is the $\kappa^{th}$ segmentation point, then we take $A_{\tau-1}(\kappa-1)$ as the $(\kappa-1)^{th}$ segmentation point  

Each computation of $\widetilde{G}_t(\kappa, w)$ needs at most $t$ evaluations of the goodness-of-fit metric. Hence there are $\mathcal{O}(T^2)$ evaluations of the goodness-of-fit metric in the $\kappa$'s iteration to obtain $\widetilde{G}_T(\kappa, w)$. This is significantly lower than the $\mathcal{O}(T^\kappa)$ evaluations prior to approximation. 

\subsection{Pruning}
\label{sec: pruning}

In the $\kappa^{th}$ iteration of dynamic programming, $A_t(\kappa, w)$ and $\widetilde{G}_t(\kappa, w)$ require searching for the optimal $\kappa^{th}$ segmentation point of $\{Z_i\}_{i=1}^t$ from candidates $\{1+\kappa^*w,\dots,t-w+1\}$. To cut computations further, we reduce this search space by only searching in $S_t(\kappa, w)$, defined below.  

For the first iteration, we initialize $S_t(1, w)=\{1+w,\dots,t-w+1\}$ which is the largest possible search space. For the $(\kappa+1)^{th}$ iteration, the search space $S_t(\kappa+1, w)$ is the result of pruning the search space $S_t(\kappa, w)$ from the previous iteration, as we want to iteratively pinpoint the most optimal change point before $t$.  The pruning rule is:
\begin{align*}
     S_t(\kappa+1, w) & = \{\tau \in S_t(\kappa, w):  \\
    &  H_t(\kappa+1, w, \tau) \geq H_t(\kappa+1, w, t-w+1)\}.
\end{align*}

In the above expression, the inequality compares the goodness-of-fit of two valid $(\kappa+1)$-segmentations for the length $t$ subsequence, one with the last change point at $\tau$ and the other at $t-w+1$. If the former is less than the latter, then $\tau$ is a less optimal segmentation point than $t-w+1$, hence $\tau$ can be pruned away from the set of candidate change points.

We choose to benchmark the goodness-of-fit induced by $\tau$ against that induced by $t-w+1$ because $t-w+1$ is the last possible change point location before $t$. Keeping $t-w+1$ in the search space would maximize the total number of change points possible for the sequence.

\subsection{Algorithm}
\label{sec: algorithm}

We outline the complete cp3o procedure in Algorithm \ref{alg: cp3o}. Aside from the notation described in Sections \ref{sec: dp} and \ref{sec: pruning}, we use $cps_t(\kappa)$ to denote the set of $\kappa$ change points estimated for the subsequence $\{Z_i\}_{i=1}^t$.

\begin{algorithm}
    \caption{cp3o}
    \label{alg: cp3o}
    \SetKwInOut{Input}{Input}
    \SetKwInOut{Output}{Output}
    \SetKwInOut{Initialize}{Initialize}

    \Input{Data sequence $z_1,z_2,\dots,z_T\in \mathbb{R}^d$ \\
    Upper bound on number of changes $K$\\
    Minimum distance between changes $w$\\}
    \Initialize{Search space $S_t(1) = \{1+w,\dots, t-w+1\}$\\
    Set of change points $cps_t(0)=\emptyset$\\
    Previous change point $A_t(0) = 1$ before $t$}
	\For {$\kappa$ from 1 to K}
		{
		\For {t from $2^*w$ to T}
			{
			$\tau^* = \argmax\limits_{\tau\in S_t(\kappa)} H_t(\kappa, w, \tau)$ \\
			$\widetilde{G}_t(\kappa,w) = H_t(\kappa, w, \tau^*)$\\			
			Update $cps_t(\kappa) = cps_{\tau^*}(\kappa-1) \cup \{\tau^*\}$	\\
			Update $A_t(\kappa) = \tau^*$ \\
			Update $S_t(\kappa+1) = \{\tau\in S_t(\kappa): H_t(\kappa+1, w, \tau) \geq H_t(\kappa+1, w, t-w+1)$
			}		
		}
	Pick optimal number of change points $\kappa^*$\\
    \Output{$cps_T(\kappa^*)$}
\end{algorithm}

The cp3o algorithm iterates through $\kappa$ from $1$ to a user-defined upper bound $K$. In each iteration $\kappa$, for each subsequence $\{Z_i\}_{i=1}^t$, cp3o finds $\tau^*$ as the $\kappa^{th}$ change point in the subsequence. The other $\kappa-1$ change points are as found in previous iterations. The goodness-of-fit $\widetilde{G}_t(\kappa,w)$ of the length $t$ subsequence is calculated with these $\kappa$ change points, and $cps_t(\kappa)$ is updated to save these change points. In preparation for future iterations, the search set $S_t(\kappa+1)$ is obtained by discarding candidate points which produce a worse goodness-of-fit than segmenting at $t-w+1$, which is the last possible segmentation point before $t$. Note that in each iteration $\kappa$, when $t=T$ is reached, we obtain estimates for $\kappa$ change points for the entire data sequence. 

We now describe the criteria for picking the optimal number of change points $\kappa^*$. Empirically, $\widetilde{G}_t(\kappa,w)$ usually increases with $\kappa$, and tends to be kinked at $k$. This is expected since partitioning beyond the optimal number of partitions should not increase the goodness-of-fit at the same rate as before. We fit a piecewise linear function with two pieces on the empirical $\widetilde{G}_t(\kappa,w)$ values, and estimate the number of change points to be the $\kappa$ at which the function transitions from one piece to the other. This is similar to techniques used to determine cutoff values from scree plots.

\section{Divergence Metrics}
\label{sec: divergence}

We now offer some guidelines for selecting the divergence metric $R$ and its sample counterpart $\widehat{R}$, then propose two metrics, namely the energy statistic and $\mathcal{A}$-distance, which satisfy these guidelines. The energy statistic and Kolmogorov-Smirnov statistic, a special case of the $\mathcal{A}$-distance, are incorporated into the cp3o procedure. We refer to the two resulting algorithms as e-cp3o and ks-cp3o.

\subsection{Selection Guidelines}
\label{sec: guidelines}

We use the notation $X \stackrel{d}{=} Y$ to mean $X$ and $Y$ are identically distributed.

\begin{property}
\label{prop: convergence} Convergence of empirical divergence to true divergence.
Let $\boldsymbol{X}_n=\{X_i\}_{i=1}^{n}$ and $\boldsymbol{Y}_m=\{Y_j\}_{j=1}^{m}$ be two sets of independent random variables;
$X_i\stackrel{d}{=}X$ for all $i$, and $Y_j\stackrel{d}{=}Y$ for all $j$.
$\widehat{R}\left(\boldsymbol{X}_n,\boldsymbol{Y}_m\right) \xrightarrow{a.s.} R\left(X,Y\right)$ as the sample size $\min\!\left(n,m\right) \rightarrow \infty$, where $R\left(X,Y\right)\geq 0$, and equality holds iff $X \stackrel{d}{=} Y$.
\end{property}

\begin{property}
\label{prop: maximizer}
Single change point detection.
For $0<\gamma<1$, suppose $Z_1,\dots,Z_{\lfloor\gamma T\rfloor} \stackrel{d}{=} X$ and $Z_{\lfloor\gamma T\rfloor+1},\dots,Z_T \stackrel{d}{=} Y$ for a sequence of any length $T$. For $0<\eta<1$, let $A\left(\eta\right)=\{Z_i\}_{i=1}^{\lfloor\eta T\rfloor}$ and $B\left(\eta\right)=\{Z_j\}_{j=\lfloor\eta T\rfloor+1}^T$.
$\widehat{R}\left(A(\eta),B(\eta)\right) \xrightarrow{a.s.} \Theta_0^1(\eta|\gamma)R(X,Y)$ as $T\to\infty$, where $\Theta_0^1(\eta|\gamma)$ maps from the interval $(0,1)$ to $\mathbb{R}$, and has a unique maximizer at $\eta=\gamma$.	
\end{property}

Property \ref{prop: convergence} concerns the convergence of the empirical divergence to the true divergence metric. It is reasonable to enforce that the non-negative divergence be $0$ when applied on two identically distributed random variables. Property \ref{prop: maximizer} implies that for a large enough sample size with one change point, the empirical divergence metric will attain its maximum value when the estimated change point location $\eta$ and true change point location $\gamma$ coincide.

\subsection{Energy Statistics}
\label{sec: estat}

The E-statistics introduces by \cite{rizzo05} are indexed by $\alpha\in(0,2)$ and allows for the detection of \textit{any} type of distributional change\footnote{If the detection of only mean changes is desired $\alpha=2$ is used.}.
For a given $\alpha$, the only distributional assumption made is that the
observations have finite absolute $\alpha^{th}$ moments.

Suppose $\boldsymbol{X}_n = \{X_i\}_{i=1}^n$ and $\boldsymbol{Y}_m = \{Y_j\}_{j=1}^m$ are iid samples from distributions with probability measures $F_X$ and $F_Y$, respectively. Then the population distance is 
\begin{equation*}
    \mathcal E(X,Y|\alpha)=2E|X-Y|^\alpha - E|X-X'|^\alpha - E|Y-Y'|^\alpha.
\end{equation*}
This is equivalent to
$$\mathcal{D}(X,Y|\alpha) = \int_{\mathbb R^d}|\phi_X(t)-\phi_Y(t)|^2\omega(t|\alpha)\,dt$$
with an appropriately chosen positive weight function $\omega$, where $\phi_X$ and $\phi_Y$ are the characteristic functions associated with distributions $F_X$ and $F_Y$, respectively.

The empirical counterpart to $\mathcal E(X,Y|\alpha)$ is
\begin{align*}
    & \widehat{\mathcal E}(\X_n,\Y_m|\alpha) =  \frac{2}{mn}\sum_{i=1}^n\sum_{j=1}^m|x_i-y_j|^\alpha \\
    & \quad - \binom{n}{2}^{-1}\!\!\!\!\sum_{1\le i<j\le n}\!\!\!\!|x_i-x_j|^\alpha
    -\binom{m}{2}^{-1}\!\!\!\!\sum_{1\le i<j\le m}\!\!\!\!|y_i-y_j|^\alpha
\end{align*}

Let $\gamma$ denote the proportion of observations from $F_X$ in the limit as $\min(n,m) \rightarrow\infty$. Then we define our divergence metrics as
\begin{align*}
R(X,Y|\alpha) &= \gamma(1-\gamma)\mathcal{E}(X,Y|\alpha), \\
\widehat{R}(\boldsymbol{X}_n,\boldsymbol{Y}_m|\alpha) &= \frac{mn}{(m+n)^2}\widehat{\mathcal E}(\X_n,\Y_m|\alpha).
\end{align*}

\begin{prop}
Properties \ref{prop: convergence} and \ref{prop: maximizer} are satisfied by the divergence metric based on the E-statistic.
\end{prop}
\begin{proof}
Using the result of \cite[Theorem 1]{matteson13} we have that 
$$\widehat R(A(\widehat{\gamma}), B(\widehat{\gamma})|\alpha)\xrightarrow{a.s.} \widehat{\gamma}(1-\widehat{\gamma})h(\widehat{\gamma};\gamma){\mathcal E}(X,Y|\alpha)$$
where $h(\widehat{\gamma};\gamma) =\left(\frac{\gamma}{\widehat{\gamma}}\mathds{1}_{\widehat{\gamma}\ge \gamma} + \frac{1-\gamma}{1-\widehat{\gamma}}\mathds{1}_{\widehat{\gamma}<\gamma}\right)^2$. Therefore, 
$R(X,Y|\alpha)=\gamma(1-\gamma){\mathcal E}(X,Y)|\alpha)$ and $\Theta_0^1(\widehat{\gamma}|\gamma)=\frac{\widehat{\gamma}(1-\widehat{\gamma})}{\gamma(1-\gamma)}h(\widehat{\gamma};\gamma)$, 
which can be shown to have a unique maximizer at $\widehat{\gamma}=\gamma$.

By definition, $\mathcal{D}(X,Y|\alpha)\geq 0$, with equality if and only if $F_X=F_Y$ by the uniqueness of characteristic functions. The rest of the proof follows from the equality of $\mathcal{D}(X,Y|\alpha)$ and $\mathcal{E}(X,Y|\alpha)$.
\end{proof}

Empirically, we use an incomplete U-statistic version of $\widehat{R}$ to reduce the number of samples needed the compute the pairwise distances. We define a window size $\delta$, within which all pairwise distances are included, and outside which only adjacent points have their pairwise distances included. 
That is,
suppose $\X_n=\{Z_a, Z_{a+1},\dots, Z_{a+n-1}\}$ and $\Y_m=\{Z_{a+n}, Z_{a+n+1},\dots, Z_{a+n+m-1}\},$ 
and define the following sets:
\begin{align*}
	&W_X^\delta = \{(i,j): a+n-\delta\le i<j<a+n\}\cup \\
	&\hspace{3cm} \bigcup_{i=0}^{n-\delta-1}\{(a+i,a+i+1)\}\\
	&W_Y^\delta = \{(i,j): a+n\le i<j<a+n+\delta\}\cup \\
	&\hspace{3cm} \bigcup_{i=\delta-1}^{m-2}\{(a+n+i,a+n+i+1)\}\\
	&B^\delta = (\{a+n-1,\dots,a+n-\delta\}\times \\ &\hspace{3cm} \{a+n,\dots,a+n+\delta-1\})\cup \\
	&\hspace{3cm} \left(\bigcup_{i=\delta+1}^{m\wedge n}\!\!\!\{(a+n-i,a+n+i-1)\}\right)
\end{align*}

The incomplete U-statistic $\widetilde{\mathcal E}$ is then
\begin{align*}
&\widetilde{\mathcal E}(\X_n,\Y_m|\alpha,\delta) = \frac{2}{\#B^\delta}\sum_{(i,j)\in B^\delta}\!\!\!\!\!|X_i-Y_j|^\alpha -\\ 
&\quad \frac{1}{\#W_X^\delta}\sum_{(i,j)\in W_X^\delta}\!\!\!\!\!|X_i-X_j|^\alpha
- \frac{1}{\#W_Y^\delta}\sum_{(i,j)\in W_Y^\delta}\!\!\!\!\!|Y_i-Y_j|^\alpha
\end{align*}

This reduces computation of $\widehat{R}(\X_n,\Y_m|\alpha)$ from $\mathcal{O}\left(n^2 \bigvee m^2\right)$ to $\mathcal{O}\left(\delta^2 + n \bigvee m\right)$. Letting $\delta\leq C\lfloor \sqrt{T} \rfloor$ for some constant $C$ results in a computational complexity of $\mathcal{O}(T)$. Note that $\delta<w$, so we set $\delta=w-1$.
It is shown 
\cite{nasari12} that a strong law of large 
numbers result holds for incomplete U-Statistics, thus the incomplete U-statistic version of $\widehat{R}$ shares the same almost sure limit as $\widehat{R}$.

\subsection{$\mathcal{A}$-distance}

The $\mathcal{A}$-distance is introduced in \cite{kifer04}. It is a generalization of the Kolmogorov-Smirnov statistic, which is often used to quantify the distance between two empirical distribution functions.

We use the same notations as in Section \ref{sec: estat}.
Let $\mathcal{A}$ be a collection of measurable sets from their domain. Then the $\mathcal{A}$-distance is defined as
$$d_\mathcal{A}(F_X,F_Y) = 2\sup_{A\in\mathcal{A}}|F_X(A)-F_Y(A)|$$ 
The empirical $\mathcal{A}$-distance is
$$\widehat{d}(\boldsymbol{X}_n,\boldsymbol{Y}_m|\mathcal{A}) = 2\sup_{A\in \mathcal{A}}\left|\frac{|\boldsymbol{x}_n \cap A|}{n} - \frac{|\boldsymbol{y}_m \cap A|}{m}\right|$$

Let $\gamma$ denote the proportion of observations from $F_X$ in the limit as $\min(n,m) \rightarrow\infty$. Then we define our divergence metrics as
\begin{align*}
    R(X,Y|\mathcal{A}) &= \gamma(1-\gamma)d_\mathcal{A}(F_X,F_Y),\\
    \widehat{R}(\boldsymbol{X}_n,\boldsymbol{Y}_m|\mathcal{A}) &= \frac{mn}{(m+n)^2}\widehat{d}(\boldsymbol{X}_n,\boldsymbol{Y}_m|\mathcal{A}).
\end{align*}
In particular, for $\mathcal{A} = \lbrace (-\infty,r) | r\in\mathbb{R} \rbrace$, $\widehat{d}(\boldsymbol{X}_n,\boldsymbol{Y}_m|\mathcal{A})$ is the Kolmogorov-Smirnov statistic.

\begin{prop}
Property \ref{prop: convergence} is satisfied by the divergence metric based on the $\mathcal{A}$-distance.
\end{prop}
\begin{proof}
We note that 
$\widehat{d}(\boldsymbol{X}_n,\boldsymbol{Y}_m|\mathcal{A}) \xrightarrow{a.s.} d_{\mathcal{A}}(F_X,F_Y)$ if $\mathcal{A}$ has a finite VC-dimension.
From \cite{kifer04}, for $M=\min(n,m)$, 
$$\!\!P[|d_A(F_X,F_Y) - \widehat{d}(\boldsymbol{X}_n,\boldsymbol{Y}_m|\mathcal{A})| \geq \epsilon] < \pi_\mathcal{A}(2M)4e^{-M\epsilon^2/4},$$ 
where for domain $D$, $\pi_\mathcal{A}(n) = \max\left\lbrace\ |\lbrace A\cap B:A\in\mathcal{A} \rbrace|: B \subseteq D \text{ and } |B| = n \right\rbrace$. For finite VC-dimension $c$, $\pi_\mathcal{A}(n)<n^c$ by Sauer's Lemma. In particular, $\mathcal{A} = \lbrace (-\infty,r) | r\in\mathbb{R} \rbrace$ has $c=2$. Hence, 
$$P[|d_A(F_X,F_Y) - \widehat{d}(\boldsymbol{X}_n,\boldsymbol{Y}_m|\mathcal{A})| \geq \upsilon] < (2M)^c 4e^{-M\upsilon^2/4}.$$

We then note that for any $\upsilon>0$, $$\sum\limits_{M=1}^\infty (2M)^c 4e^{-M\upsilon^2/4} = 4(2^c)Li_{-c}\left(e^{-\upsilon^2/4}\right) < \infty$$ where $Li_{-c}(x)$ is the polylogarithm function. Hence, $\widehat{d}(\boldsymbol{X}_n,\boldsymbol{Y}_m|\mathcal{A}) \xrightarrow{a.s.} d_{\mathcal{A}}(F_X,F_Y) \geq 0$, and if $F_X=F_Y$, then $\widehat{d}(\boldsymbol{X}_n,\boldsymbol{Y}_m|\mathcal{A}) \xrightarrow{a.s.} 0$.

The proof concludes with noticing $\frac{n}{m+n}\rightarrow \gamma$.
\end{proof}

\begin{prop}
	Property \ref{prop: maximizer} is satisfied by the divergence metric based on the $\mathcal{A}$-distance.
\end{prop}
\begin{proof}
As $\min(n,m)\rightarrow\infty$,
$$\widehat{R}(A(\widehat{\gamma}),B(\widehat{\gamma})|\mathcal{A}) \xrightarrow{a.s.} \widehat{\gamma}\left( 1-\widehat{\gamma} \right) g\left(\widehat{\gamma};\gamma\right) d_\mathcal{A}(F_X,F_Y)$$ 
where $g\left(\widehat{\gamma};\gamma\right)=\left( \frac{\gamma}{\widehat{\gamma}}\mathds{1}_{\widehat{\gamma}\geq\gamma} + \frac{1-\gamma}{1-\widehat{\gamma}}\mathds{1}_{\widehat{\gamma}<\gamma} \right)$.
Therefore $\Theta_0^1(\widehat{\gamma}|\gamma) = \frac{\widehat{\gamma}\left(1-\widehat{\gamma}\right)}{\gamma(1-\gamma)}g\left(\widehat{\gamma};\gamma\right) = 
\left( \frac{1-\widehat{\gamma}}{1-\gamma}\mathds{1}_{\widehat{\gamma}\geq\gamma} + \frac{\widehat{\gamma}}{\gamma}\mathds{1}_{\widehat{\gamma}<\gamma} \right)$, and it is maximized at $\widehat{\gamma}=\gamma$.
\end{proof}

\section{Simulation Study}
\label{sec: simulation}

To assess the performance of the segmentations, we use Fowlkes and Mallows' adjusted Rand index \citep{fowlkes83}. This value is calculated by 
comparing an estimated segmentation to the true segmentation. The index takes into account both the number of change points as well as their 
locations, and lies in the interval $[0,1]$, where it is equal to $1$ if and only if the two segmentations are identical.

We also include two measures of discrepancy between the true and estimated change point locations as an assessment of estimation accuracy. T2E is the average shortest distance from a true change point to the estimated change points, and E2T is the average shortest distance from an estimated change point to the true change points. A low T2E shows that all the true change points are well-estimated, and a low E2T shows that all the estimated change points are close to true change points.

For each simulation scenario, we apply various methods to 100 randomly generated time series, each with three evenly-spaced change points. 
We compare our methods with E-divisive \citep{ecp14} and NPCP-F \citep{holmes13} (nonparametric, approximate/bisection search), and PELT \citep{killick12} (parametric, exact/dynamic programming search).
All methods were run with their default parameter values unless otherwise specified. For E-Divisive and e-cp3o this corresponds 
to $\alpha=1$. For PELT, e-cp3o and ks-cp3o, the upper bound of number of changes $K$ was set to $5$. For all methods, the minimum segment size was set to approximately $1.5\sqrt{T}$ observations, that is, $w=30,60,90,120$ for time series of length $T=400,1600,3200,6000$, respectively.
All experiments were run on a standard desktop computer.

\subsection{Effects of Pruning}

We demonstrate the effects of the pruning step within the dynamic program on the search space $S_t(\kappa)$ in Figure \ref{fig: heatmap}. The darker the color, the bigger the search space. The search space is pruned significantly within a few iterations.

\begin{figure}
    \centering
    \includegraphics[width=3.25in]{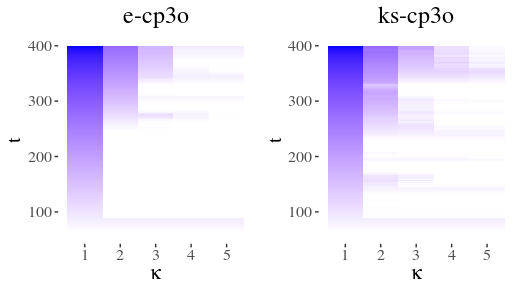}
    \caption{Color represent number of candidate change points in search space $S_t(\kappa)$ at each time index $t$ and iteration $\kappa$. The darker the color, the higher the number. $S_t(\kappa)$ is at it's maximum at $\kappa=1$, and is rapidly pruned in subsequent iterations.}
    \label{fig: heatmap}
\end{figure}

\subsection{Simulation 1}
\label{sec: unisim1}

\begin{figure*}
    \centering
    \begin{subfigure}[b]{\linewidth}
        \includegraphics[width=\linewidth]{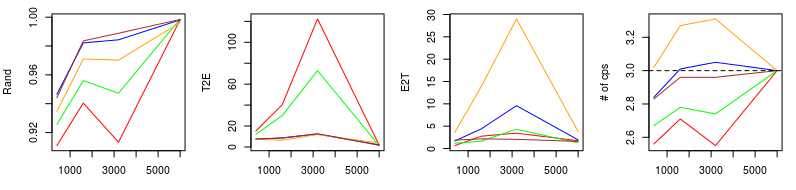}
        \caption{Simulation 1 with mean and variance changes in Gaussian distribution.}
        \label{fig:meanSim}
    \end{subfigure}
    \hspace{\fill}
    \begin{subfigure}[b]{\linewidth}
        \includegraphics[width=\linewidth]{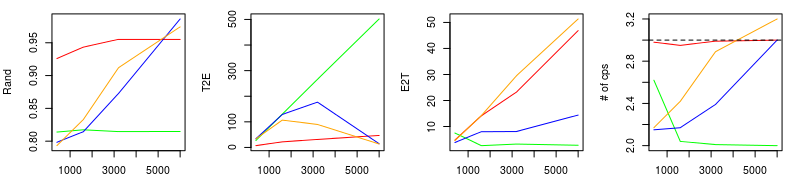}
        \caption{Simulation 2 with distribution, mean and tail changes.}
        \label{fig:tailSim}
    \end{subfigure}
    \hspace{\fill}
	\begin{subfigure}[b]{\linewidth}
       \includegraphics[width=\linewidth]{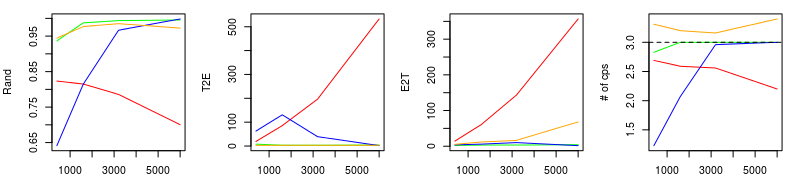}
       \caption{Simulation 3 with changes in t and Cauchy distribution.}
       \label{fig:uniSim}
    \end{subfigure}
	\begin{subfigure}[b]{\linewidth}
	\centering
       \includegraphics[width=0.5\linewidth]{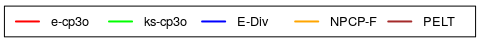}
    \end{subfigure}    
    \caption{Average Rand index, discrepancy values and number of change points detected against length of time series. True number of change points denoted by black dotted line. Good performance is reflected by Rand close to 1, small T2E and E2T, and estimated change point number close to 3.}
    \label{fig:plotSim}
\end{figure*}

This set of simulations consist of independent Gaussian observations which undergo changes in their mean and variance. 
The distribution parameters were chosen so that $\mu_j\stackrel{iid}{\sim}Unif(-10,10)$ and $\sigma^2_j\stackrel{iid}{\sim}Unif(0,5)$. 

\begin{table}
\centering
\resizebox{\columnwidth}{!}{
\begin{tabular}{cccccc}
\hline
\multicolumn{1}{c}{T} & \multicolumn{1}{c}{e-cp3o} & \multicolumn{1}{c}{ks-cp3o} & \multicolumn{1}{c}{E-Div} & \multicolumn{1}{c}{NPCP-F} & \multicolumn{1}{c}{PELT}\\ \hline
400     & 0.119 & 0.918 & 5.285 & 5.475 & 0.002\\ 
1600    & 1.811 & 77.869 & 118.672 & 85.288 & 0.020\\
3200    & 7.977 & 683.003 & 756.503 & 346.632 & 0.087 \\
6000    & 27.855 & 4951.490 & 1841.570 & 1357.562 & 0.342\\\hline
\end{tabular}
}
\caption{\label{table:meanSim}
Average runtimes (s) of the first univariate simulation from Section \ref{sec: unisim1} with mean and variance changes in Gaussian distributions.
}
\end{table}

As can be seen from Table \ref{table:meanSim} and Figure \ref{fig:meanSim}, PELT was fast and suffered little loss in accuracy in identifying change points in longer time series, as observed from the Rand and discrepancy (T2E and E2T) values. 
At $T=400,1600$ and $3200$, from the lower estimated number of change points and the higher values of T2E, we notice that e-cp3o and ks-cp3o did not always detect all the changes. But from the lower values of E2T, we see that the points which e-cp3o and ks-cp3o did identify as changes are amongst the closest to the true changes. At $T=6000$, e-cp3o and ks-cp3o performed comparably with the competing methods in terms of segmentation quality.

In the {\tt ecp} {\it R} package, we also provide a faster version of ks-cp3o which only computes the Kolmogorov-Smirnov statistic using points within a window of size $\delta$ around each candidate segmentation point. Its average runtime at $T=6000$ is 14.483$s$, but it detected a slightly lower average number of change points (2.620), and hence is excluded from the reporting.  

\subsection{Simulation 2}
\label{sec: unisim2}

Time series in this simulation study contain a change in distribution, mean, and tail index.
The data transitions from a exponential distribution $Exp\left(\frac{1}{3}\right)$ to a normal distribution $N(3,1)$ to a standard normal distribution $N(0,1)$.
The tail index change is caused by a final transition to a t-distribution with $2.01$ 
degrees of freedom.

We do not include PELT in the following experiments since it is a parametric method that detects only mean and variance changes.

We expect that all methods included will be able to easily detect the mean change and will have more difficulty detecting the change in tail index. Results for this set of simulations can be found in Figure \ref{fig:tailSim}. Runtimes are similar to those in Table \ref{table:meanSim} with PELT excluded.

At $T=400, 1600$ and $3200$, e-cp3o was not only significantly faster than all other procedures, but also managed to generate the best segmentations on average. While most procedures tended to miss the tail index change, e-cp3o detected the most number of change points with averages within $0.05$ of the true number $3$. e-cp3o had higher E2T since it picked out the third change point more often than the other methods, but the accuracy of detecting the third change was not as high as those for the first two changes. At $T=6000$, E-Divisive overtook in terms of segmentation quality, but e-cp3o was much faster and hence provided a better balance between speed and accuracy.

\subsection{Simulation 3}
\label{sec: unisim3}

The data transitions from a t-distribution $t_{0.1}$ to $t_{1.9}$ to a Cauchy distribution $Cauchy(-2,1)$ to $Cauchy(0,1)$. We use $\alpha=0.09$ instead since we need $\alpha<0.1$ for the moment assumptions of E-statistics to hold. Complete results are shown in Figure \ref{fig:uniSim}. Runtimes are similar to those in Table \ref{table:meanSim} with PELT excluded.

In the short time series setting ($T=400$), NPCP-F and ks-cp3o performed comparatively. In the long time series setting ($T=6000$), E-Divisive and ks-cp3o performed comparatively. In general, ks-cp3o had the most consistent performance by almost always achieving the highest Rand and lowest discrepancy values. In fact, ks-cp3o picked out the correct number of change points in every sample series from $T=1600$ onwards. It demonstrated great potential in change point detection in general datasets where commonly desired distributional properties cannot be assumed.

Due to the small value of $\alpha$ which makes the E-statistics smaller in magnitude and therefore more difficult to distinguish,  e-cp3o does not perform as well as the other methods. Moreover, it is not straightforward to determine the best $\alpha$ to use in practice, especially when extreme observations are present. Hence, it is important to select a goodness-of-fit that is appropriate for the data.

\section{Applications to Real Data}

\subsection{Temperature Anomalies}

We examine the HadCRUT4 dataset of \cite{morice12}. This dataset consists of monthly global temperature anomalies from 1850 to 
2014. We perform analysis on the land air temperature anomaly component from the tropical region, and apply change point procedures to the differenced data which visually appears to be piecewise stationary.

The e-cp3o and ks-cp3o procedures were applied with a minimum segment length of $5$ years, corresponding to $w=60$; a maximum of $K=5$ change points were fit. We chose default values $\alpha=1$ for e-cp3o. e-cp3o identified change points at April 1917 and April 1969, and ks-cp3o identified changes at February 1864, May 1878 and September 1898. These are shown in Figure 
\ref{fig:climateTS}.

\begin{figure}[h!]
    \centering
    \includegraphics[width=0.9\linewidth]{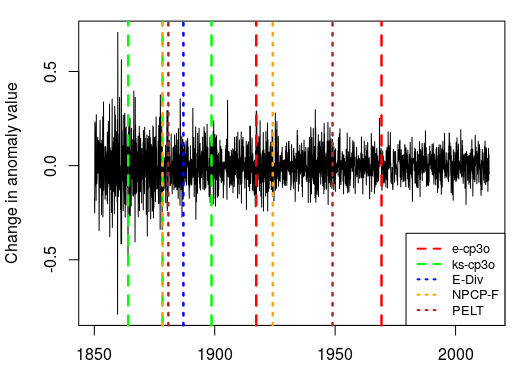}
    \caption{Change in land air temperature anomalies for the Tropical climate zone from February 1850 to December 2013. Estimated change point locations indicated by dashed vertical lines.}
    \label{fig:climateTS}
\end{figure}

The May 1878 change point may be a result of a large climate disruption in 1877-1878, which may be caused by a major El Ni\~{n}o episode. The April 1969 change point occurs around the United Nations Conference on the Human Environment. This conference, which was held in June 1972, 
focused on human interactions with the environment. 

e-cp3o used $3.659s$ and ks-cp3o used $376.138s$. With the same parameters, competing methods E-divisive, NPCP-F and PELT used $135.795s$, $111.377s$ and $0.078s$ respectively. Segmentation results vary and the true change points are unknown, which make it difficult to compare methods. However, we note that even though ks-cp3o took the longest time, it is the only method that identified the 1864 change point, which visually does look like a true change.

\subsection{Exchange Rates}

We apply e-cp3o to a set of spot foreign exchange (FX) rates obtained through the \texttt{R} package \texttt{Quandl} \citep{quandl13}, and compare results with multivariate methods E-divisive and NPCP-F. We consider the 3-dimensional time series consisting of monthly FX rates for Brazil (BRL), Russia (RUB), and Switzerland (CHF) against the United States (USD). The time 
horizon spanned is September 30, 1996 to February 28, 2014, which results in a total of 210 observations. We look at the change in the log rates, such that marginal processes appear to be piecewise stationary.\par

The e-cp3o procedure is applied with a minimum segment length of $12$ observations (a year), which corresponds to a value of $w=12$. Furthermore, we have chosen to fit at 
most $K=5$ change points, and default values of $\alpha=1$ is used. This specific choice of values resulted in change points being identified at December 1998, August 2002 and April 2008. These results are depicted in Figure \ref{fig:FXrates}.

\begin{figure*}[h!]
    \centering
    \begin{subfigure}[b]{.3\linewidth}
        \includegraphics[width=\linewidth]{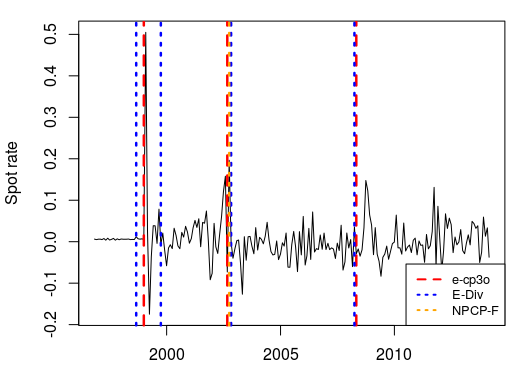}
        \caption{Brazil}
        \label{fig:braFX}
    \end{subfigure}
    \hspace{\fill}
    \begin{subfigure}[b]{.3\linewidth}
        \includegraphics[width=\linewidth]{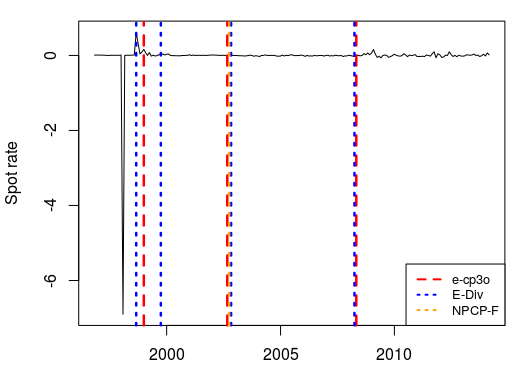}
        \caption{Russia}
        \label{fix:swsFX}
    \end{subfigure}
    \hspace{\fill}
	\begin{subfigure}[b]{.3\linewidth}
       \includegraphics[width=\linewidth]{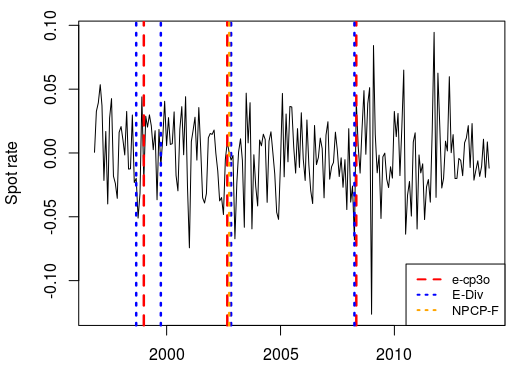}
       \caption{Switzerland}
       \label{fig:rusFX}
    \end{subfigure}
    \caption{Time series for FX spot rates for each of the three countries' currencies versus the USD. Estimated change point locations indicated by dashed vertical lines.}
    \label{fig:FXrates}
\end{figure*}

Changes in Russia's economic standing leading up to the 1998 ruble crisis may be the cause of the December 1998 change point. The August 2002 and April 2008 change points may be the results of the 2002 South American economic crisis and 2008 financial crisis respectively.

e-cp3o used $0.026s$, while E-Divisive and NPCP-F used $1.59s$ and $1.213s$ respectively using the same parameters. e-cp3o and E-Divisive found similar change points, but e-cp3o is much faster. NPCP-F seemed to have underestimated the number of change points and only identified one whereas the other methods identified three or more.

\section{Conclusion}

We have presented an approximate search procedure that incorporates pruning in order to reduce the amount of unnecessary calculations and to dramatically reduce computational costs. This search method 
can be used with almost any goodness-of-fit metric in order to identify change points in univariate and multivariate time series. In addition, this is accomplished 
without the user having to specify any sort of penalty parameter or function.

As the simulation studies demonstrate, the cp3o procedures do not uniformly record the best running time, average Rand values or average discrepancy values. However, when accuracy and computation time are viewed together across different data scenarios, the cp3o procedures are either better or comparable to almost all other competitors. Moreover, greater care in choosing a goodness-of-fit metric that is suitable to the data application is likely to improve performance further in terms of accuracy and/or speed. Hence, we would advocate the cp3o procedure as a general purpose change point procedure.

\section*{Acknowledgement}

Matteson was supported by an NSF CAREER award (DMS-1455172), a Xerox PARC Faculty Research Award, and the Cornell University Atkinson Center for a Sustainable Future (AVF-2017).


\bibliographystyle{IEEEtranN}
\bibliography{IEEEabrv,sample-bibliography}

\end{document}